\documentclass[letterpaper, 10pt, twocolumn]{article}
\usepackage{titlesec}
\titleformat{\section}{\normalfont\large\bfseries}{\thesection}{1em}{} 
\titleformat{\subsection}{\normalfont\normalsize\bfseries}{\thesubsection}{1em}{} 
\titleformat{\subsubsection}{\normalfont\small\bfseries}{\thesubsubsection}{1em}{} 
\usepackage[letterpaper,top=0.75in,bottom=1in,left=0.625in,right=0.625in]{geometry}
\usepackage{cite}
\usepackage{amsmath,amssymb,amsfonts}
\usepackage{graphicx}
\usepackage{textcomp}

\usepackage{mathtools}
\usepackage{relsize}
\usepackage{tikz}
\usetikzlibrary{positioning, arrows.meta}

\usepackage{tabularx} 
\usepackage{graphicx} 
\usepackage{cite} 
\usepackage[final]{hyperref} 
\usepackage{amsmath} 
\usepackage{amssymb}  
\usepackage{mathtools, cuted,bm}
\usepackage{etoolbox}
\patchcmd{\proof}{\indent}{}{}{}
\usepackage{caption} 
\usepackage{tabularray}

\usepackage{tabularx}
\usepackage{array}
\newcolumntype{Y}{>{\centering\arraybackslash}X}

\usepackage[justification=centering]{caption}

\newcommand{\twodots}{.\kern-0.1em.}
\newcommand{\myldots}{\kern-0.05em.\kern-0.01em.\kern-0.01em.\kern0.01em}

\usepackage{caption}
\usepackage{subcaption}
\usepackage{esvect}
\usepackage{lipsum, color}

\newcommand{\R}{\mathbb{R}}

\newcommand{\fe}{\mathsf{f}}
\newcommand{\ve}{\mathsf{v}}

\usepackage{relsize}
\usepackage{tikz}
\graphicspath{{./Figures/}} 

\usepackage[margin=0.8cm]{caption}
\usepackage{pgfplots}
\pgfplotsset{compat=newest}
\usepackage{tabularx}
\pgfplotsset{plot coordinates/math parser=false}

\usepackage{ulem}
\usepackage{cancel}

\usepackage{booktabs}

\usepackage{amsthm}
\newtheorem{corollary}{Corollary}
\newtheorem{remark}{Remark}

\newtheorem{proposition}{Proposition}

\newcommand{\norm}[1]{\left\lVert#1\right\rVert}

\makeatletter
\def\munderbar#1{\underline{\sbox\tw@{$#1$}\dp\tw@\z@\box\tw@}}
\makeatother

\definecolor{olivegreen}{RGB}{128, 128, 0}

\usepackage{algpseudocode}
\usepackage{algorithm}

\usepackage{caption}

\newcommand\scalemath[2]{\scalebox{#1}{\mbox{\ensuremath{\displaystyle #2}}}}

\hypersetup{
	colorlinks=true,       
	linkcolor=black,        
	citecolor=black,        
	filecolor=magenta,     
	urlcolor=black         
}

\usepackage{amsthm}

\usepackage{cite}
\usepackage{amsmath,amssymb,amsfonts}
\usepackage{graphicx}
\usepackage{textcomp}

\usepackage{mathtools}
\usepackage{relsize}
\usepackage{tikz}
\usepackage{threeparttable}
\usepackage{siunitx}

\usepackage{tabularx} 
\usepackage{graphicx} 
\usepackage{cite} 
\usepackage[final]{hyperref} 
\usepackage{amsmath} 
\usepackage{amssymb}  
\usepackage{mathtools, cuted,bm}
\usepackage{etoolbox}
\patchcmd{\proof}{\indent}{}{}{}
\usepackage{caption} 
\usepackage{tabularray}

\usepackage[justification=centering]{caption}

\usepackage{caption}
\usepackage{subcaption}
\usepackage{esvect}
\usepackage{lipsum, color}

\newcommand{\zs}{z^{\mathrm{s}}}
\newcommand{\vs}{v^{\mathrm{s}}}
\newcommand{\yr}{y^{\mathrm{r}}}
\newcommand{\uc}{u^{\mathrm{c}}}
\newcommand{\up}{u^{\mathrm{p}}}
\newcommand{\Xs}{X^{\mathrm{s}}}
\newcommand{\xr}{\mathbf{x}^{\mathrm{r}}}

\usepackage{relsize}
\usepackage{tikz}
\graphicspath{{./Figures/}} 
\usepackage[margin=0.8cm]{caption}
\usepackage{pgfplots}
\pgfplotsset{compat=newest}
\usepackage{tabularx}
\pgfplotsset{plot coordinates/math parser=false}

\makeatletter
\def\munderbar#1{\underline{\sbox\tw@{$#1$}\dp\tw@\z@\box\tw@}}
\makeatother

\definecolor{olivegreen}{RGB}{128, 128, 0}

\usepackage{algorithm}
\usepackage{algpseudocode}

\usepackage{caption}

\usepackage{relsize}
\usepackage{tikz}
\usepackage[margin=0.1cm]{caption}
\usepackage{pgfplots}
\pgfplotsset{compat=newest}
\usepackage{tabularx}
\pgfplotsset{plot coordinates/math parser=false}
\usepackage{booktabs}

\usepackage{ulem}
\usepackage{cancel}

\makeatletter
\def\munderbar#1{\underline{\sbox\tw@{$#1$}\dp\tw@\z@\box\tw@}}
\makeatother
\definecolor{olivegreen}{RGB}{128, 128, 0}

\usepackage{algpseudocode}
\usepackage{algorithm}

\hypersetup{
	colorlinks=true,       
	linkcolor=black,        
	citecolor=black,        
	filecolor=magenta,     
	urlcolor=black         
}

\usepackage{algorithm}
\usepackage{algpseudocode}

\newtheorem{theorem}{Theorem}
\usepackage{threeparttable}

\usepackage{comment}
\newtheorem{lemma}[theorem]{Lemma}

\begin{document}
	\title{\textbf{Dual MPC for quasi-Linear Parameter Varying systems
	}}
	
	\author{Sampath Kumar Mulagaleti and Alberto Bemporad, \textit{Fellow, IEEE}
		\thanks{This work was funded by the European Union (ERC Advanced Research Grant COMPACT, No. 101141351). Views and opinions expressed are however those of the authors only and do not necessarily reflect those of the European Union or the European Research Council. Neither the European Union nor the granting authority can be held responsible for them.
			The authors are with the IMT School for Advanced Studies Lucca, Italy. (Email: \url{s.mulagaleti@imtlucca.it},\url{alberto.bemporad@imtlucca.it})}
	}


	\newcounter{tempEquationCounter}
	\newcounter{thisEquationNumber}
	\newenvironment{floatEq}
	{\setcounter{thisEquationNumber}{\value{equation}}\addtocounter{equation}{1}
		\begin{figure*}[!t]
			\normalsize\setcounter{tempEquationCounter}{\value{equation}}
			\setcounter{equation}{\value{thisEquationNumber}}
		}
		{\setcounter{equation}{\value{tempEquationCounter}}
			\hrulefill\vspace*{4pt}
		\end{figure*}
		
	}
	\newenvironment{floatEq2}
	{\setcounter{thisEquationNumber}{\value{equation}}\addtocounter{equation}{1}
		\begin{figure*}[!t]
			\normalsize\setcounter{tempEquationCounter}{\value{equation}}
			\setcounter{equation}{\value{thisEquationNumber}}
		}
		{\setcounter{equation}{\value{tempEquationCounter}}
		\end{figure*}
	}

	\maketitle
	
	\begin{abstract}
		We present a dual Model Predictive Control (MPC) framework for the simultaneous identification and control of quasi-Linear Parameter Varying (qLPV) systems. The framework is composed of an online estimator for the states and parameters of the qLPV system, and a controller that leverages the estimated model to compute inputs with a dual purpose: tracking a reference output while actively exciting the system to enhance parameter estimation. The core of this approach is a robust tube-based MPC scheme that exploits recent developments in polytopic geometry to guarantee recursive feasibility and stability in spite of model uncertainty. The effectiveness of the framework in achieving improved tracking performance while identifying a model of the system is demonstrated through a numerical example.
	\end{abstract}
	
	\thispagestyle{empty}
	\pagestyle{empty}
	
\section{Introduction}
The closed-loop performance of a Model Predictive Control (MPC) scheme relies heavily on the quality of the model used. When dealing with minor prediction inaccuracies, robust MPC formulations provide an effective mechanism to ensure constraint satisfaction and stability \cite{rawlings2017model}. However, as the complexity of control systems grows, the dynamical models employed are often highly nonlinear and identified directly from data \cite{Bemporad2025}. For such systems, characterizing global uncertainty bounds is often not viable, and even if characterized, a robust MPC scheme using them might be excessively conservative. To ameliorate this conservativeness, it becomes necessary to adapt the model online utilizing input-output data generated in closed loop. Typically, such adaptive MPC schemes \cite{Clarke1996} couple an online parameter estimation mechanism with an MPC formulation endowed with stability guarantees in the presence of model changes.

However, passive adaptation can be significantly enhanced by actively generating control inputs that simultaneously improve the quality of the estimated model. This principle forms the core of dual control, where the applied input actively excites the system to facilitate parameter estimation~\cite{feldbaum1960dual1,Bar-Shalom1974,Mesbah2018}. In the context of linear systems with input-state measurements, this is generally achieved by ensuring the computed input is persistently exciting~\cite{Marafioti2013,Lorenzen2019,Cannon2023}, or directly minimizing the prediction uncertainty~\cite{Heirung2015}. For nonlinear systems with input-state measurements, bounded-uncertainty estimators are typically employed alongside active exploration strategies to shrink parameter uncertainty sets~\cite{Guay2009,Kohler2021,Berberich2025,Soloperto2023,Baltussen2025}.

When restricted to only input-output measurements, the dual-control problem becomes markedly more challenging \cite{Heirung2018}. A prevalent strategy in such settings involves using an Extended Kalman Filter (EKF) to perform joint state-parameter estimation along with their associated uncertainty, which are then propagated within the MPC framework to simultaneously minimize uncertainty alongside the control objective, and ensure probabilistic constraint satisfaction \cite{Lucia2014,Houska2017,Mesbah2018,Messerer2023}. While effective, propagating this uncertainty inherently yields a nonlinear and nonconvex optimization problem, necessitating the development of specialized solvers. Furthermore, unlike approaches that assume full state feedback, these output-feedback schemes typically lack guarantees regarding closed-loop stability in the presence of online parameter variations.

\textbf{Contribution}. We present an approach to synthesize dual MPC schemes for quasi-Linear Parameter Varying (qLPV) systems, in which the scheduling function is restricted to the simplex using a $\mathrm{softmax}$ operator. This parameterization effectively models complex nonlinear dynamics \cite{mulagaleti2025}, while being suitable for the synthesis of robust MPC schemes since the nonlinearity is bounded by construction. We estimate online the state and parameters of the model using an estimator that is subject to constraints. These constraints are designed to ensure that the downstream dual MPC scheme, developed using a tube-based MPC (TMPC) approach, remains feasible. We solve the dual MPC problem in two stages, with the first stage computing a \textit{tracking} control input to drive the model to a neighborhood of an output reference, and the second stage a \textit{perturbing} input to promote active exploration. The bounds on this perturbation are user-specified, which is a standard exploration-exploitation tradeoff in dual MPC. To develop the TMPC scheme, we extend the approach of \cite{Badalamenti2024,Badalamenti2025} to compute configuration-constrained polytopic tubes \cite{Villanueva2024}. We demonstrate that the closed-loop is input-to-state (ISS) Lyapunov stable \cite{Limon2006} against parameter variations, and validate its effectiveness using a simple numerical example. 

\textbf{Notation}. We denote by $\mathbb{N}$ the set of natural numbers. For $n \in \mathbb{N}$, we define $\Delta_n:=\left\{x \in \R^{n} \middle| x \in [0,1], \sum_{i=1}^n x_i = 1\right\}$ as the simplex set.
We denote by $\mathbb{I}_n^m$ the set of integers between $n$ and $m$, and $\mathbf{I}_n$ as the identity matrix in $\R^{n \times n}$.
We define $\mathrm{softmax}(x):\R^{n} \to \R^{n}$ as the function with components $e^{x_i}/\sum_{j=1}^n e^{x_j}$ for $i \in \mathbb{I}_1^n.$ Given a vector $a \in \R^n$, $|a| \in \R^n$ denotes the element-wise absolute value vector. Given compact convex sets $M_1,\cdots,M_N \subseteq \R^{n}$, we define $\mathrm{CH}\{M_i,i \in \mathbb{I}_1^N\}:=\left\{\sum_{i=1}^N \lambda_i x_i \in \R^{n} \middle| x_i \in M_i, \lambda \in \Delta_N\right\}$ as the convex hull of the individual sets, and $M_i \oplus M_j \subseteq \R^{n}$ as the Minkowski sum of the sets $M_i$ and $M_j$. If $M_i=\{x\}$, i.e., its a singleton, then we denote the sum as $x \oplus M_j$.

\section{Problem Setup}
We want to control the nonlinear discrete-time dynamical system
\begin{align}
	\label{eq:plant}
	\mathbf{z}^+ = \mathbf{f}(\mathbf{z},u), && y=\mathbf{g}(\mathbf{z})
\end{align}
where $u \in \R^{n_u}$, $y \in \R^{n_y}$ and $\mathbf{z} \in \R^{n_{\mathbf{z}}}$ are the input, output, and state vectors
at time $t\in\mathbb{N}$, respectively, and $\mathbf{z}^+$ denotes the state at the next time step, under the following setting: $(a)$ The functions $\mathbf{f}$ and $\mathbf{g}$, the state dimension $n_{\mathbf{z}}$ are unknown; $(b)$  The functions $\mathbf{f}$ and $\mathbf{g}$ are continuously differentiable, and satisfy $0=\mathbf{f}(0,0)$ and $0=\mathbf{g}(0)$; $(c)$ The input is subject to given constraints $u \in \mathbb{U}:=\{u:|u|\leq \epsilon^u\}$; $(d)$ Only the output $y$ is measurable, and is subject to given convex constraints $y \in \mathbb{Y}$.

Under the above assumptions, we aim to design an output-reference tracking controller for the plant in a \textit{dual-control} framework, in which we identify a model of the plant~\eqref{eq:plant} while controlling it by applying a suitable input $u$. To achieve this goal, our dual-control scheme consists of the following components: $(a)$  A dynamical model of the plant with state $x \in \R^{n_x}$ and parameters $\theta \in \R^{n_{\theta}}$; $(b)$ A dual-objective MPC scheme that at each time $t$ uses the current state and parameter estimates to compute an input $u_t \in \mathbb{U}$ which balances output reference tracking with active exploration; $(c)$ A mechanism to estimate the state $x_{t+1}$ and parameters $\theta_{t+1}$ using the output $y_{t+1}$ of the plant.
We endow this closed-loop scheme with recursive feasibility and asymptotic stability guarantees using modern computational tools from polytopic geometry.
\section{Modelling}
We model the dynamics of \eqref{eq:plant} using the qLPV system
\begin{align}
	\label{eq:model}
	x^+ = A(p(x,u))x+B(p(x,u))u, && \hat{y}=Cx,
\end{align}
where $x \in \R^{n_x}$, $C \in \R^{n_y \times n_x}$, and the matrix-valued functions $(A(p),B(p))$ are parameterized as
\begin{align}
	\label{eq:LP}
	(A(p),B(p)):=\sum_{i=1}^{n_p} p_i(A_i,B_i),
\end{align}
with $A_i \in \R^{n_x \times n_x}$ and $B_i \in \R^{n_x \times n_u}$. We parameterize the scheduling function $p:\R^{n_x+n_u} \to \R^{n_p}$ as
\begin{align}
	\label{eq:softmax}
	p(x,u) = \mathrm{softmax}(\mathcal{N}(x,u)),
\end{align}
where $\mathcal{N}: \R^{n_x+n_u} \to \R^{n_p}$ is a feedforward neural network (FNN) with continuously differentiable activation units. The parameterization in \eqref{eq:softmax} results in $ p(x,u) \in \Delta_{n_p}$,
i.e., the scheduling variable always belongs to the simplex, and is hence bounded by construction.
We denote by $\theta \in \R^{n_{\theta}}$ the parameters of \eqref{eq:model}, which include the system matrices $\{A_i,  B_i\}$ and FNN parameters, and assume instead that $C$ is fixed and is a full row-rank matrix, such as a collection of $n_y$ rows of the identity matrix of order $n_x$, with $n_y\leq n_x$. 

\subsection{Combined state and parameter estimation}
We recall the general approach to estimate the state $x$ and parameters $\theta$ of \eqref{eq:model} using input-output measurements from the plant. Denoting $f(x,u,\theta):=A(p(x,u))x+B(p(x,u))u$, we introduce the augmented system
\begin{align}
	\label{eq:model_EKF}
	\zeta^+ = F(\zeta,u)+w, && y=\tilde{C} \zeta +v,
\end{align}
where $\zeta=(x,\theta)$ is the augmented state with dynamics $F(\zeta,u)=(f(x,u,\theta),\theta)$ and process noise $w \in \R^{n_x+n_{\theta}}$,
the output matrix is $\tilde{C}=[C \ 0]$, and output noise $v \in \R^{n_y}$. Using input-output measurements $u_{0:t}:=(u_0,\cdots,u_t)$ and $y_{0:t+1}:=(y_0,\cdots,y_{t+1})$, an observer estimates the state and parameter of \eqref{eq:model} at time $t+1$ as
\begin{align}
	\label{eq:general_observer}
	\hat{\zeta}_{t+1} = \hat{F}(u_{0:t},y_{0:t+1}).
\end{align}
For example, when using an Extended Kalman Filter (EKF), the process and measurement noise are assumed to satisfy $w \sim \mathcal{N}(0,Q_{\mathrm{e}})$ and $v \sim \mathcal{N}(0,R_{\mathrm{e}})$ respectively. Then, assuming that $\zeta_t \sim \mathcal{N}(\hat{\zeta}_t,\hat{P}_t)$, the mean and covariance of the state of \eqref{eq:model_EKF} are updated as%
\begin{subequations}
	\label{eq:EKF_steps}
	\begin{align}
		&\tilde{\zeta}_{t+1}=F(\hat{\zeta}_t,u_t),  \label{eq:EKF_steps:P1} \\
		&\tilde{P}_{t+1} = \nabla_{\zeta}F(\hat{\zeta}_t,u_t)\hat{P}_t\nabla_{\zeta}F(\hat{\zeta}_t,u_t)^{\top}+Q_{\mathrm{e}}, \label{eq:EKF_steps:P2} \\
		&K_{t+1}=\tilde{P}_{t+1}\tilde{C}^{\top}(\tilde{C}\tilde{P}_{t+1}\tilde{C}^{\top}+R_{\mathrm{e}})^{-1}, \label{eq:EKF_steps:P3}\\
		&\hat{\zeta}_{t+1}=\tilde{\zeta}_{t+1}+K_{t+1}(y_{t+1}-\tilde{C}\tilde{\zeta}_{t+1}), \label{eq:EKF_steps:P4}\\
		& \hat{P}_{t+1}=(\mathbf{I}_{n_x+n_{\theta}}-K_{t+1}\tilde{C})\tilde{P}_{t+1}. \label{eq:EKF_steps:P5}
	\end{align}
\end{subequations}
Alternative approaches, e.g., the Moving Horizon Estimator (MHE) \cite{Rao2003,Muller2023} can be used to model \eqref{eq:general_observer}. 
In this work, we assume that given a constraint set $\Theta \subseteq \R^{n_x+n_{\theta}}$, the estimator is \eqref{eq:general_observer} is designed to satisfy
\begin{align}
	\label{eq:constraint_theta}
	\hat{\zeta}_{t+1} \in \Theta.
\end{align}
While the MHE can explicitly handle such constraints, the EKF can be modified \cite{Simon2010,Ansari2017} to handle them by modifying the correction step in \eqref{eq:EKF_steps:P4} as
\begin{align}
	\label{eq:correction_qp_con}
	\hat{\zeta}_{t+1} = \scalemath{0.93}{ \arg\min_{\zeta \in \Theta} \|\zeta-\tilde{\zeta}_{t+1}\|_{\tilde{P}_{t+1}^{-1}}^2+\|y_{t+1}-\tilde{C}\zeta\|_{R_{\mathrm{e}}^{-1}}^2.}
\end{align}
In the next section, we formulate the constraint set $\Theta$ to guarantee stability of the proposed dual-control scheme.

\section{Dual MPC}
Towards developing a dual-control framework, we first formulate the MPC problem: 
\begin{align}
	&\min_{\mathbf{v},\vs} \sum_{k=0}^{N-1} \norm{\begin{bmatrix}z_k-\zs \\ v_k-\vs \end{bmatrix}}_Q^2+\norm{\begin{bmatrix}z_N - \zs  \\ v_N -\vs  \end{bmatrix} }_P^2 \label{eq:NMPC} \\ 
	& \hspace{120pt} + \ell(\vs,\yr) + \alpha a(\mathbf{v}) \nonumber \\
	& \ \text{s.t.} \ z_{k+1} = f(z_k,v_k,\hat{\theta}), \ \zs = f(\zs,\vs,\hat{\theta}), \nonumber \\
	& \hspace{15pt} \ Cz_k \in \mathbb{Y}, \ v_k \in \mathbb{U}, \ k \in \mathbb{I}_0^{N-1}, \ \nonumber\\
	& \hspace{15pt} \  C\zs \in \mathbb{Y}, \ \vs \in \mathbb{U}, \ z_0 = \hat{x}, \ z_N \in \mathcal{O}(\vs,\hat{\theta}), \nonumber
\end{align}
which depends on the state $\hat{x}$ and parameters $\hat{\theta}$ of \eqref{eq:model} estimated using \eqref{eq:general_observer}, and the output reference  $\yr \in \R^{n_y}$. The optimization vector is the finite sequence of moves $\mathbf{v}=(v_0,\cdots,v_N)$. 
Problem \eqref{eq:NMPC} is formulated using an \textit{artificial reference} approach \cite{Krupa2024}, in which the steady-state $\zs$ corresponding to an input $\vs$ is optimized online based on $\yr$ through $\ell(\vs,\yr)$, such as $\ell(\vs,\yr)=\|\yr-C\zs\|_2^2$. Furthermore, the function $a(\mathbf{v})$ is the \textit{active exploration} objective, designed to modify the input such that  exploration is promoted, e.g., as in \cite{Feng2018}. The factor $\alpha \geq 0$ weighs the active exploration objective against the tracking objective. Finally, the cost function matrices $Q,P \succeq 0$, along with the terminal set $\mathcal{O}(\vs,\hat{\theta})$ with terminal control input $v_N$ are designed to ensure recursive feasibility and stability of the closed-loop scheme formulated with $u_t=v_0^*(\hat{x}_t,\hat{\theta}_t,\yr_t)$, where $v_0^*(\hat{x}_t,\hat{\theta}_t,\yr_t)$ is the first input of the parametric optimizer of Problem \eqref{eq:NMPC}.
Unfortunately, solving Problem \eqref{eq:NMPC} entails the following challenges:
\begin{itemize}
	\item It requires the use of NLP solvers.
	\item Recursive feasibility and stability cannot be guaranteed if the parameter $\hat{\theta}$ is updated using \eqref{eq:general_observer}.
\end{itemize}
We now present a TMPC approximation of Problem \eqref{eq:NMPC} that addresses both these challenges. 

\begin{remark}
	\label{remark:constraints}
	Problem \eqref{eq:NMPC} is formulated using a \textit{certainty-equivalence} approach similar to \cite{Heirung2015}.
	While this nominal formulation does not explicitly account for the effect of unmodeled dynamics on output constraint satisfaction, it serves as a theoretical baseline. In practice, prediction uncertainties are typically addressed via constraint tightening~\cite{Yan2002, Messerer2023, Köhler2019, Farina2015}. Rather than removing output constraints to trivially guarantee recursive feasibility, we retain them in this nominal setting. This establishes the stability properties of the dual-control scheme, and provides the framework necessary for integration with constraint-tightening techniques.
\end{remark}

\subsection{Tube-based MPC}
We construct our TMPC scheme based on the following observations and modifications:

\textbf{Perturbed input:} If the initial model $\hat{\theta}_0$ is \textit{good enough}, a practical implementation of \eqref{eq:NMPC} would require a small active exploration constant $\alpha \geq 0$ in the objective. Then, under sufficient smoothness of the dynamics and cost functions, the optimal solution to Problem \eqref{eq:NMPC} is continuous with respect to the weighting parameter $\alpha$. Consequently, for small values of $\alpha$, the optimal dual-control solution lies in a bounded neighborhood of the pure-tracking solution, i.e., with $\alpha=0$.
Motivated by this intuition, we propose a two-stage strategy. Given a tuning parameter $\beta \in [0,1)$, we split the control input as $u=\uc+\up$ with $\uc \in (1-\beta) \mathbb{U}$ and $\up \in \beta \mathbb{U}$. In the first stage, we design a robust tracking controller to compute $\uc$ using the uncertain model
\begin{align}
	\label{eq:lpv_additive}
	x^+ \in A(p(x,u))x+B(p(x,u))\uc \oplus W,
\end{align}
where the additive disturbance $W$ is defined as
\begin{align}
	\label{eq:W_defn}
	w \in W:=\mathrm{CH}\left\{\beta B_i\mathbb{U}, i \in \mathbb{I}_1^{n_p}\right\}.
\end{align}
In the second stage, we utilize an active exploration criterion to compute the perturbation $\up \in \beta \mathbb{U}$, following which we apply the input $u=\uc+\up$ to the plant. Future study can analyze the suboptimality of this approach.

\textbf{Model encapsulation:} From parameterizations \eqref{eq:LP} and \eqref{eq:softmax}, we observe that for any $(x,u) \in \R^{n_x} \times \R^{n_u}$, the propagated state $x^+$ according to \eqref{eq:lpv_additive} satisfies
\begin{align}
	\label{eq:multiplicative}
	x^+ \in \mathrm{CH}\left\{A_ix + B_i \uc \oplus W, i \in \mathbb{I}_1^{n_p}\right\}
\end{align}
since $p(x,u) \in \Delta_{n_p}$. Then, a \textit{tube} constructed for the difference inclusion \eqref{eq:multiplicative} will encapsulate the set-valued trajectories of \eqref{eq:lpv_additive}. Formally, for any $M \in \mathbb{N}$, a tube is a sequence of sets $\{X_0,\cdots,X_M\}\subseteq \R^{n_x}$ satisfying
\begin{align}
	\label{eq:tube_defn}
	& \hspace{2pt} \forall x \in X_t, \ \exists \uc \in (1-\beta)\mathbb{U} : \\
	& \hspace{20pt}  A_i x+B_i\uc \oplus W \subseteq X_{t+1}, \ \  \forall \ i \in \mathbb{I}_1^{n_p}, \ t \in \mathbb{I}_0^{M-1}. \nonumber
\end{align}
This condition implies that if $x_0 \in X_0$, then there exists an input sequence $\{\uc_0, \cdots,\uc_{M-1}\} \in (1-\beta)\mathbb{U}$ such that the resulting state trajectory of \eqref{eq:lpv_additive} satisfies $x_t \in X_t$ for $t \in \{1,\cdots,M\}$ for any $\{w_0,\cdots,w_{M-1}\} \in W$.

TMPC schemes construct a feasible tube that converges to a \textit{robust control invariant} (RCI) set $\Xs \subseteq \R^{n_x}$ which satisfies the inclusion $C\Xs \subseteq \mathbb{Y}$ along with
\begin{align}
	\label{eq:RCI}
	& \hspace{-0pt} \forall x \in \Xs, \ \exists \uc \in (1-\beta)\mathbb{U} : \\
	& \hspace{25pt}  A_i x+B_i\uc \oplus W \subseteq \Xs, \ \  \forall \ i \in \mathbb{I}_1^{n_p}. \nonumber
\end{align}
\subsection{Configuration-constrained TMPC}
We develop our TMPC scheme using polytopes
\begin{align}
	X \leftarrow X(z,s):=z \oplus \left\{x \in \R^{n_x} \middle| Fx \leq s\right\},
\end{align}
where the matrix $F \in \R^{\fe \times n_x}$ is fixed a priori, and the parameters optimized online are the center $z \in \R^{n_x}$ and offset vector $s \in \R^{\fe}$. We also enforce nonnegative and configuration constraints
\cite{Villanueva2024} on $s$:
\begin{align}
	\label{eq:cc_constraint}
	\mathcal{E}:=\left\{s \in \R^{\fe} \middle|  s \geq 0, Es \leq 0\right\},
\end{align}
which ensures $0 \in \left\{x \in \R^{n_x} \middle| Fx \leq s\right\}$, and preserves the combinatorial structure of $X(z,s)$. Then, there exist matrices $V:=\{V_j \in \R^{n_x \times \fe},j \in \mathbb{I}_1^{\ve}\}$ such that for any $z \in \R^{n_x}$,
\begin{align}
	\label{eq:vertex}
	\scalemath{0.95}{s \in \mathcal{E} \Rightarrow X(z,s)=\mathrm{CH}\left\{z+V_j s, j \in \mathbb{I}_1^{\ve}\right\}.}
\end{align}
Thus, the set $\mathcal{E}$ thus enables a linear parameterization of the halfspace and vertex representations of $X(z,s)$.
\begin{proposition}
	\label{prop:RCI_result}
	The polytope $\Xs = X(\zs,s)$ satisfies the RCI condition in \eqref{eq:RCI} if there exist vectors $\vs \in \R^{n_u}$, $c \in \R^{\ve n_u}$ and $q \in\R^{\fe}$ verifying the inequalities%
	\begin{align}
		\label{eq:RCI_cc}
		&F(A_i(\zs+V_js)+B_i (\vs+U_jc))+d+q \leq s+F\zs,  \nonumber \\
		&C(\zs+V_js) \in \mathbb{Y}, \ \vs+U_jc \in (1-\beta) \mathbb{U}, \ q \geq 0, \nonumber \\
		& s \in \mathcal{E}, \ \forall (i,j) \in \mathbb{I}_1^{n_p} \times \mathbb{I}_1^{\ve},
	\end{align}
	where we define the disturbance vector
	\begin{align}
		\label{eq:d_defn}
		d := \max \{\beta|FB_i|\epsilon^u,i \in \mathbb{I}_1^{n_p}\}
	\end{align}
	and matrices $U_j:=e_j \otimes \mathbb{I}_{\ve} \in \R^{n_u \times \ve n_u}$.
\end{proposition}
\begin{proof}
	The result follows from \cite[Corollary 4]{Villanueva2024}, where $c=(c_1,\cdots,c_{\ve})$ stacks the vertex control inputs, and \eqref{eq:RCI_cc} contains an additional disturbance vector $q \geq 0$.
\end{proof}

Given model parameters $\hat{\theta}$ estimated by \eqref{eq:general_observer} and output reference $\yr \in \R^{n_y}$, the optimal RCI set is defined as $X(\zs_{\mathrm{o}}(\hat{\theta},\yr),s_{\mathrm{o}}(\hat{\theta},\yr))$, where $(\zs_{\mathrm{o}},s_{\mathrm{o}})$ are part of the optimizers of the problem
\begin{align}
	\label{eq:optimal_rci_problem}
	r(\hat{\theta},\yr):=\min_{\xr} \ell(\xr,\yr) \ \text{s.t.} \ \ \eqref{eq:RCI_cc}
\end{align}
defined over $\xr=(\zs,\vs,s,c,q) \in \R^{n_x+n_u+2m+\ve n_u}$. The optimal RCI cost is modeled, for example, as
\begin{align}
	\label{eq:ell_formulation}
	\ell(\xr,\yr)=\sum_{j=1}^{\ve}\|\yr-C\zs\|_{Q_1}^2+\norm{\xr}_{Q_2}^2,
\end{align}
with $Q_1,Q_2 \succ 0$, such that the center of the optimal RCI set projected onto the output space tracks the reference. Our goal is to formulate a TMPC scheme which computes a tube $\{X(z_t,s),t \in \mathbb{N}\}$ that converges to $X(\zs_{\mathrm{o}}(\hat{\theta},\yr),s_{\mathrm{o}}(\hat{\theta},\yr))$.

\begin{proposition}
	\label{prop:tube_result}
	Suppose the vector $\xr$ satisfies \eqref{eq:RCI_cc}. Then, the sets $\scalemath{0.95}{\{X(z_0,s),\cdots,X(z_M,s)\}}$ satisfy the tube condition \eqref{eq:tube_defn} if there exist vectors $\{v_0,\cdots,v_{M-1}\} \in \R^{n_u}$ satisfying
	\begin{align}
		\label{eq:tube_cc}
		&F(A_i(z_t-\zs)+B_i(v_t-\vs))\leq q+F(z_{t+1}-\zs), \nonumber \\
		&C(z_t+V_js) \in \mathbb{Y}, \ v_t+U_jc \in (1-\beta)\mathbb{U}, \nonumber \\
		& \forall (i,j,t) \in \mathbb{I}_1^{n_p} \times \mathbb{I}_1^{\ve} \times \mathbb{I}_0^{M-1}.
	\end{align}
\end{proposition}
\begin{proof}
	From \cite[Corollary 4]{Villanueva2024}, the sets $X(z_t,s)$ and $X(z_{t+1},s)$ are part of an invariant tube if
	\begin{align}
		\label{eq:tube_to_hold}
		\hspace{-5pt} \scalemath{0.98}{F(A_i(z_t+V_js)+B_i (v_t+U_jc))+d \leq s+Fz_{t+1}}
	\end{align}
	holds with some $v_t \in \R^{n_u}$ for all $(i,j) \in \mathbb{I}_1^{n_p} \times \mathbb{I}_1^{\ve}$, along with the input and output constraints. From \eqref{eq:RCI_cc}, the inequality $\scalemath{0.95}{
		F(A_iV_js+B_iU_jc) \leq F\zs+s-(F(A_i\zs+B_i\vs)+d+q)}$
	holds. Then, adding $F(A_iz_t+B_iv_t)+d$ to both sides of this inequality, we see that the left-hand-side of \eqref{eq:tube_to_hold} is $\leq F\zs+s+F(A_i(z_t-\zs)+B_i(v_t-\vs))-q$. The proof concludes by enforcing this bound to be $\leq s+Fz_{t+1}$.
\end{proof}

Following Propositions \ref{prop:RCI_result} and \ref{prop:tube_result}, we formulate our TMPC controller based on the QP
\begin{align}
	&\hspace{-3pt} \min_{\mathbf{x}} \scalemath{0.88}{\sum_{k=0}^{N-1} \norm{\begin{bmatrix}z_k-\zs \\ v_k-\vs \end{bmatrix}}_Q^2+\norm{\begin{bmatrix}z_N - \zs  \\ v_N -\vs  \end{bmatrix} }_P^2} + \ell(\mathbf{x}^{\mathrm{r}},\yr)\label{eq:TMPC} \\
	& \ \text{s.t.} \ \scalemath{0.95}{F(A_i(z_k-\zs)+B_i(v_k-\vs))\leq q+F(z_{k+1}-\zs),} \nonumber \\
	& \hspace{15pt} \ \scalemath{0.95}{C(z_k+V_js) \in \mathbb{Y}, \ v_k+U_jc \in (1-\beta)\mathbb{U}, } \nonumber \\
	& \hspace{15pt} \ \scalemath{0.95}{F(A_i(z_N-\zs)+B_i(v_N-\vs))\leq q+\gamma F(z_N-\zs),} \nonumber \\
	& \hspace{15pt} \ \scalemath{0.95}{C(z_N+V_js) \in \mathbb{Y}, \ v_N+U_jc \in (1-\beta)\mathbb{U}, } \nonumber \\
	&  \hspace{15pt} \ F\hat{x} \leq q+Fz_0, \ \eqref{eq:RCI_cc}, \ (i,j,k) \in \mathbb{I}_1^{n_p} \times \mathbb{I}_1^{\ve} \times \mathbb{I}_0^{N-1}, \nonumber
\end{align}
where $\gamma \in (0,1)$ is some user-specified constant. The QP is parametric in the state $\hat{x}$ and parameter $\hat{\theta}$ of \eqref{eq:model} estimated using \eqref{eq:general_observer}, along with reference $\yr$. The optimization vector is $\mathbf{x}=(z_0,v_0,\cdots,z_N,v_N,\xr)$.
Denoting the parametric optimizer of \eqref{eq:TMPC} as $\mathbf{x}^*(\hat{x},\hat{\theta},\yr)$, the closed-loop scheme is defined as%
\begin{subequations}
	\label{eq:closed_loop}
	\begin{align}
		&u_t =v^*_0(\hat{x}_t,\hat{\theta}_t,\yr_t)+\sum_{j=1}^{\ve}  \lambda_j c_j^*(\hat{x}_t,\hat{\theta}_t,\yr_t)+\up_t, \label{eq:closed_loop:1} \\
		&(\hat{x}_{t+1},\hat{\theta}_{t+1}) \ \text{updated using \eqref{eq:general_observer} subject to \eqref{eq:constraint_theta}}, \label{eq:closed_loop:2}
	\end{align}
\end{subequations}
for some $\up_t \in \beta \mathbb{U}$ computed using an active exploration criterion, where $\lambda \in \R^{\ve}$ is computed as
\begin{align}
	\label{eq:lambda_problem}
	\min_{\lambda \in \Delta_{\ve}} \ \|\lambda\|_2^2 \ \  \text{s.t.} \ \ \hat{x}=z_0^*+\sum_{j=1}^{\ve} \lambda_j V_js^*.
\end{align}
\subsection{Recursive feasibility and stability}
To guarantee recursive feasibility when the state and parameters are updated via \eqref{eq:general_observer}, we restrict the estimator's update step in \eqref{eq:constraint_theta}. To this end, we define a parameterized polytope $\Theta(\mathbf{x}, d)$ to formulate \eqref{eq:constraint_theta} as
\begin{align*}
	\Theta(\mathbf{x},d):=\left\{\begin{pmatrix} x \\ \theta \end{pmatrix} \middle|
	\scalemath{0.92}{
		\begin{matrix*}[l]
			Fx \leq q+Fz_1, \ \beta|FB_i|\epsilon^u \leq d,  \\
			F(A_i(z_k-\zs)+B_i(v_k-\vs)),  \\
			\hspace{60pt} \leq q+F(z_{k+1}-\zs), \\
			F(A_i(z_N-\zs)+B_i(v_N-\vs) \\
			\hspace{60pt}  \leq q+ F(z^+-\zs), \\
			F(A_i(z^+-\zs)+B_i(v^+-\vs) \\ 
			\hspace{60pt}  \leq q+ \gamma F(z^+-\zs), \\
			F(A_i(\zs+V_js)+B_i (\vs+U_jc)) \\
			\hspace{60pt}  \leq s+F\zs-(d+q), \\
			\forall \ (i,j,k) \in \mathbb{I}_1^{n_p} \times \mathbb{I}_1^{\ve} \times \mathbb{I}_1^{N-1},
	\end{matrix*}}
	\right\}
\end{align*}
where $z^+$, $v^+$ are the shifted variables
\begin{align}
	\label{eq:feasible_seq}
	z^+=\zs+\gamma(z_N-\zs), \ v^+=\vs+\gamma(v_N-\vs).
\end{align}
\begin{lemma}
	\label{lemma:recursive_feasibility}
	Suppose Problem \eqref{eq:TMPC} is feasible with initial state-parameter estimate $(\hat{x}_0,\hat{\theta}_0)$. Then, \eqref{eq:TMPC} is feasible for all $t \in \mathbb{N}$ with $(\hat{x}_t,\hat{\theta}_t)$ generated by the closed-loop scheme in \eqref{eq:closed_loop} for any reference sequence $\{\yr_t,t \in \mathbb{N}\}$ and perturbation sequence $\{\up_t \in \beta \mathbb{U},t \in \mathbb{N}\}$ when \eqref{eq:constraint_theta} is defined as
	\begin{align}
		\label{eq:EKF_constraint}
		(\hat{x}_{t+1},\hat{\theta}_{t+1}) \in \Theta(\mathbf{x}^*(\hat{x}_t,\hat{\theta}_t,\yr_t),d(\hat{\theta}_t))
	\end{align}
	with the vector $d(\hat{\theta}_t)$ is computed as in \eqref{eq:d_defn}.
\end{lemma}
\begin{proof}
	We split the proof into two parts. First, we show that recursive feasibility holds if $\hat{\theta}_{t+1}=\hat{\theta}_t$, i.e., the parameters are held constant, and the state is updated as $\hat{x}_{t+1}=f(\hat{x}_t,u_t,\hat{\theta}_t)$. Then, we extend the proof when they are updated following \eqref{eq:closed_loop:2}. We denote the optimizer $\mathbf{x}^*(\hat{x}_t,\hat{\theta}_t,\yr_t)$ of Problem \eqref{eq:TMPC} at time $t$ as $\mathbf{x}^*=(z^*_0,v^*_0,\cdots,z^*_N,v^*_N,\mathbf{x}^{\mathrm{r}*})$.
	
	Suppose \eqref{eq:TMPC} is feasible at time $t \in \mathbb{N}$ and $u_t \in \mathbb{U}$ from \eqref{eq:closed_loop:1} is applied with some $\up_t \in \beta \mathbb{U}$. From \eqref{eq:tube_to_hold} in Proposition \ref{prop:tube_result}, we know that for all $i \in \mathbb{I}_1^{n_p}$, the inequality $F(A_i(z_0^*+V_js^*)+B_i (v_0^*+U_jc^*+\up_t)) \leq s^*+Fz_1^*$
	holds for any $\up_t \in \beta \mathbb{U}$. Denoting $p_t = p(\hat{x}_t,u_t)$, \eqref{eq:softmax} implies $p_t \in \Delta_{n_p}$. Hence, multiplying both sides of the inequality with $p_{t,i}$ and components $\lambda_j$ computed as in \eqref{eq:lambda_problem} and summing, $A(p_t)\hat{x}_t+B(p_t)u_t \in X(z_1^*,s^*)$ follows. Hence, the constraint $F\hat{x}\leq s+Fz_0$ in \eqref{eq:TMPC} is feasible with $\hat{x}=A(p_t)\hat{x}_t+B(p_t)u_t $ and $z_0 = z^*_1$. Furthermore, the RCI constraints in \eqref{eq:RCI_cc} are feasible with $\xr=\mathbf{x}^{\mathrm{r}*}$ and $\hat{\theta}_{t+1}=\hat{\theta}_t$ for any $\yr_{t+1}$, since $\yr$ only enters through the objective. Finally, we consider the candidate sequence%
	\begin{subequations}
		\label{eq:candidates}
		\begin{align} 
			z_k &= z_{k+1}^*, \ v_k = v_{k+1}^*, \ k \in \mathbb{I}_0^{N-1},  \label{eq:candidates:zv} \\ 
			z_N &= z_N^*+\gamma(z_N^*-z^{\mathrm{s}*}), \ v_N = v_N^*+\gamma(v_N^*-v^{\mathrm{s}*}) \label{eq:candidates:z} 
		\end{align}
	\end{subequations}
	at time $t+1$. Clearly, \eqref{eq:candidates:zv} satisfies the constraints for $k \in \mathbb{I}_0^{N-2}.$ For $k=N$, we know that the inequality
	\begin{align}
		\label{eq:terminal_inequality}
		\scalemath{0.95}{F(A_i(z_N^*-z^{\mathrm{s}*})+B_i(v_N^*-v^{\mathrm{s}*})\leq q^*+\gamma F(z_N^*-z^{\mathrm{s}*})}
	\end{align}
	holds since \eqref{eq:TMPC} is feasible at time $t$. Hence, the inequality $\scalemath{0.95}{F(A_i(z_{N-1}-z^{\mathrm{s}})+B_i(v_{N-1}-v^{\mathrm{s}})\leq q+F(z_N-z^{\mathrm{s}})}$
	holds at time $t+1$ with $k=N-1$ from \eqref{eq:candidates:z} along with $(z_{N-1},v_{N-1},\zs,\vs,q)=(z_N^*,v_N^*,z^{\mathrm{s}*},v^{\mathrm{s}*},q^*)$. Finally, multiplying both sides of \eqref{eq:terminal_inequality} with $\gamma \in (0,1)$, and noting that $\gamma q^*\leq q^*$ since $q^* \geq 0$, the inequality $F(A_i(z_N-z^{\mathrm{s}})+B_i(v_N-v^{\mathrm{s}})\leq q+\gamma  F(z_N-z^{\mathrm{s}})$ follows.
	Thus, \eqref{eq:TMPC} is recursively feasible with $\hat{x}_{t+1}=f(\hat{x}_t,u_t,\hat{\theta}_t)$ and $\hat{\theta}_{t+1}=\hat{\theta}_t$. 
	To conclude the proof when $\hat{x}_{t+1}$ and $\hat{\theta}_{t+1}$ are updated as per \eqref{eq:closed_loop:2}, note that the subsequent state and parameters satisfy \eqref{eq:EKF_constraint}, with the set $\Theta(\mathbf{x}^*,d(\hat{\theta}_t))$ defined using the feasible sequence in \eqref{eq:candidates}. Hence, for any $(\hat{x}_{t+1},\hat{\theta}_{t+1}) \in \Theta(\mathbf{x}^*,d(\hat{\theta}_t))$, \eqref{eq:candidates} is a feasible sequence at time $t+1$.
\end{proof}

In the following result on stability of the closed-loop scheme, we assume a constant reference $\yr_t\equiv\yr$, $\forall t \in \mathbb{N}$, and hence drop its time dependence.
\begin{theorem}
	\label{thm:stability}
	Suppose Problem \eqref{eq:TMPC} is feasible with initial state and parameter estimate $(\hat{x}_0,\hat{\theta}_0)$, and \eqref{eq:constraint_theta} is such that \eqref{eq:EKF_constraint} holds. Suppose further that matrices $Q,P \succ 0$ are chosen such that $P \succeq (1-\gamma^2)^{-1}Q$, and denote the optimal value of \eqref{eq:TMPC} as $\mathcal{C}(\hat{x},\hat{\theta})$. Then, $\mathcal{L}(\hat{x},\hat{\theta}):=\mathcal{C}(\hat{x},\hat{\theta})-r(\hat{\theta})$ serves as an ISS-Lyapunov function for \eqref{eq:closed_loop} with respect to the optimal RCI set against $\|\hat{\theta}_{t+1}-\hat{\theta}_t\|$ irrespective of the perturbation sequence $\{\up_t \in \beta \mathbb{U},t \in \mathbb{N}\}$ in \eqref{eq:closed_loop:1}, where $r(\hat{\theta})$ is the optimal RCI cost from \eqref{eq:optimal_rci_problem}.
\end{theorem}
\begin{proof}
	From Lemma \ref{lemma:recursive_feasibility}, we know that \eqref{eq:TMPC} remains recursively feasible if \eqref{eq:EKF_constraint} holds. Let us call $\mathbf{x}^*=\mathbf{x}^*(\hat{x}_t,\hat{\theta}_t)$ its optimizer. Firstly, observe that the tracking components of the objective of \eqref{eq:TMPC} are nonnegative since $Q,P \succ 0$, and $\ell(\mathbf{x}^{\mathrm{r}*}) \geq r(\hat{\theta}_t)$ since $r(\hat{\theta}_t)$ is the optimal value of $\ell(\mathbf{x}^{\mathrm{r}})$ subject to \eqref{eq:RCI_cc}. Hence, $\mathcal{L}(\hat{x},\hat{\theta}) > 0$ whenever $(z^*_0,v^*_0) \neq (z^{\mathrm{s}}_{\mathrm{o}}(\hat{\theta}),v^{\mathrm{s}}_{\mathrm{o}}(\hat{\theta}))$, or equivalently if $\hat{x} \notin X(z^{\mathrm{s}}_{\mathrm{o}}(\hat{\theta}),s_{\mathrm{o}}(\hat{\theta}))$.
	Substituting the sequence in \eqref{eq:candidates} and denoting the deviation $m_{k}=(z^*_k-z^{\mathrm{s}*},v^*_k-v^{\mathrm{s}*})$ for $k \in \mathbb{I}_0^{N}$, it follows that
	\begin{align}
		\label{eq:Lyapunov_ineq_1}
		& \mathcal{L}(\hat{x}_{t+1},\hat{\theta}_{t+1})-\mathcal{L}(\hat{x}_{t},\hat{\theta}_{t}) \leq \\
		& \hspace{20pt}  -\|m_0\|_Q^2 +\|m_N\|_{Q+\gamma^2P-P}^2+r(\hat{\theta}_{t+1})-r(\hat{\theta}_t).  \nonumber
	\end{align}
	Clearly, $P \succeq (1-\gamma^2)^{-1}Q$ implies that the second term is nonpositive. Furthermore, observe that $\hat{\theta}$ enters the constraints \eqref{eq:RCI_cc} bilinearly with the primal variables in the strongly convex QP \eqref{eq:optimal_rci_problem}, such that its optimal value $r(\hat{\theta})$ is Lipschitz continuous in $\hat{\theta}$ \cite{Fiacco1990} under Linear Independence Constraint Qualification. Hence, there exists some  $L_r>0$ such that $ r(\hat{\theta}_{t+1})-r(\hat{\theta}_t) \leq L_r \|\hat{\theta}_{t+1}-\hat{\theta}_t\|$, and then
	\begin{align*}
		\scalemath{0.95}{\mathcal{L}(\hat{x}_{t+1},\hat{\theta}_{t+1})-\mathcal{L}(\hat{x}_{t},\hat{\theta}_{t}) \leq  -\|m_0\|_Q^2 + L_r \|\hat{\theta}_{t+1}-\hat{\theta}_t\|}
	\end{align*}
	holds, concluding the proof.
\end{proof}

From Theorem \ref{thm:stability}, we see that if \eqref{eq:general_observer} results in bounded updates such that $\|\hat{\theta}_{t+1}-\hat{\theta}_t\|<\infty$, then the cost increase is guaranteed to be bounded. For the EKF in \eqref{eq:EKF_steps}, this is guaranteed if the covariance matrix $\hat{P}_t$ remains bounded, which is expected when the linearized augmented system is uniformly observable \cite{Reif1999}. 
The convergence of the estimation scheme for the qLPV model class in the presence of time-varying constraints in \eqref{eq:EKF_constraint} is a subject of future research. 

\subsection{Active exploration}
\label{sec:persistence_excitation}
From Lemma \ref{lemma:recursive_feasibility}, we know that for any perturbation input $\{\up_t \in \beta \mathbb{U},t \in \mathbb{N}\}$ defining the control law in \eqref{eq:closed_loop:1}, \eqref{eq:TMPC} remains feasible. Accordingly, we define the problem
\begin{align}
	\label{eq:AL_problem}
	& \max_{\mathbf{u}} \ a(\mathbf{u}) \\
	& \ \ \text{s.t.} \ \ x_{k+1}=f(x_k,u_k,\hat{\theta}), \ x_0 = \hat{x}, \nonumber \\
	& \hspace{22pt} \ x_{k+1} \in X(z_{k+1}^*,s^*), \ u_k \in \mathbb{U}, \  k \in \mathbb{I}_0^{N^{\mathrm{p}}-1} \nonumber 
\end{align}
that is parametric in the current state and parameter $(\hat{x},\hat{\theta})$ estimated by \eqref{eq:general_observer}, along with the optimizer $\hat{\mathbf{x}}^*=\hat{\mathbf{x}}^*(\hat{x},\hat{\theta},\yr)$ of \eqref{eq:TMPC}. For some user-specified $N^{\mathrm{p}} \in \mathbb{N}$, the tube sequence in the constraints is defined with $s^*=s^*(\hat{x},\hat{\theta},\yr)$, $z_k^*=z_k^*(\hat{x},\hat{\theta},\yr)$ if $k \in \mathbb{I}_1^N$ and 
\begin{align}
	\label{eq:extended_tube}
	z^*_{k+1}=\gamma z^*_k+(1-\gamma)z^{\mathrm{s}*}(\hat{x},\hat{\theta},\yr)
\end{align}
otherwise. The active exploration objective $a(\mathbf{u})$ is defined over $\mathbf{u}=(u_0,\cdots,u_{N^{\mathrm{p}}-1})$. Some examples include the persistence of excitation criterion, e.g., \cite{Matej2013,Marafioti2013}; Fisher information, e.g., \cite{Lucia2014}; Predicted covariance minimization, e.g., \cite{Hovd2004,Heirung2015,Houska2017}; General nonlinear regression-type criteria, e.g., \cite{xie2025}. Defining $\mathbf{u}^*(\hat{x},\hat{\theta},\yr)$ as the optimizer of Problem \eqref{eq:AL_problem}, the closed-loop scheme is defined as
\begin{subequations}
	\label{eq:closed_loop_dual}
	\begin{align}
		&u_t =\mathbf{u}^*_0(\hat{x}_t,\hat{\theta}_t,\yr_t) \label{eq:closed_loop_dual:1} \\
		&(\hat{x}_{t+1},\hat{\theta}_{t+1}) \ \text{updated using \eqref{eq:general_observer} subject to \eqref{eq:EKF_constraint}}. \label{eq:closed_loop_dual:2}
	\end{align}
\end{subequations}
\begin{corollary}
	\label{corr:feasible}
	Suppose that Problem \eqref{eq:TMPC} with the initial state and parameter estimate $(\hat{x}_0,\hat{\theta}_0)$ is feasible. Then, \eqref{eq:TMPC}, \eqref{eq:AL_problem} and the constrained estimator \eqref{eq:general_observer} subject to \eqref{eq:EKF_constraint} remain recursively feasible irrespective of the reference $\{\yr_t,t \in \mathbb{N}\}$, and ISS-Lyapunov stable if the reference is held constant.
\end{corollary}
\begin{proof}
	The proof follows from Lemma \ref{lemma:recursive_feasibility}, with the tube sequence for $k \geq N$ in \eqref{eq:extended_tube} satisfying \eqref{eq:tube_to_hold} from \eqref{eq:terminal_inequality} and \eqref{eq:RCI_cc}. Then, stability follows from Theorem \ref{thm:stability}.
\end{proof}
Note that in \eqref{eq:AL_problem}, optimizing the input sequence online freely might be computationally expensive. To overcome this difficulty, we approximately solve \eqref{eq:AL_problem} using a sample-based approach. We first sample the set of perturbation trajectories
\begin{align}
	\label{eq:sample_perturbations}
	\mathrm{U}^{\mathrm{p}}:=\{\mathrm{u}^{\mathrm{p}}_j=(\up_0,\cdots,\up_{N^{\mathrm{p}}-1}) \in \beta \mathbb{U}, j \in \mathbb{I}_1^{M^{\mathrm{p}}}\},
\end{align}
and, for each $\mathrm{u}^{\mathrm{p}} \in  \mathrm{U}^{\mathrm{p}}$, we compute the corresponding nominal input sequence as $ \uc_k = v^*_k+\sum_{j=1}^{\ve}  \lambda_j c_j^*$,
with $v^*_k$ defined similarly as \eqref{eq:extended_tube} for $k>N$ and $\lambda$ computed as in \eqref{eq:lambda_problem} with $(\hat{x},z_0^*) \leftarrow (x_k,z_k^*)$, with the state $x_k$ propagated using the input $u_k=\uc_k+\up_k$. Evaluating the maximizing sequence, we apply the corresponding input to the plant.

\section{Numerical example}
For the purpose of numerical illustration, we consider the nonlinear mass-spring-damper system with dynamics{\footnote{Code to reproduce the results is found on \url{https://github.com/samku/DMPC_qLPV}}} 
\begin{align*}
	\mathrm{m}_{[1]}\ddot{\mathrm{x}}_{[1]}&=10 u-\mathrm{k}^{\mathrm{s}}(\mathrm{x}_{[1]})-\mathrm{k}^{\mathrm{d}}(\dot{\mathrm{x}}_{[1]})-\mathrm{k}^{\mathrm{s}}(\delta \mathrm{x})-\mathrm{k}^{\mathrm{d}}(\dot{ \delta \mathrm{x}}), \\
	\mathrm{m}_{[2]}\ddot{\mathrm{x}}_{[2]}&=-\mathrm{k}^{\mathrm{s}}(\mathrm{x}_{[2]})-\mathrm{k}^{\mathrm{d}}(\mathrm{x}_{[2]})+\mathrm{k}^{\mathrm{s}}(\delta \mathrm{x})+\mathrm{k}^{\mathrm{d}}(\dot{ \delta \mathrm{x}}),
\end{align*}
where $(\mathrm{x}_{[1]},\mathrm{x}_{[2]})$ are the positions of the masses. The input $u$ is the force applied on the first mass, and the output is the position of the second mass. Denoting $\delta \mathrm{x}=\mathrm{x}_{[1]}-\mathrm{x}_{[2]}$, the reaction forces are $\mathrm{k}^{\mathrm{s}}(x)=\mathrm{a}x+\mathrm{b}x^3$ and $\mathrm{k}^{\mathrm{d}}(v)=\mathrm{d}v+\mathrm{e} \cdot \mathrm{tanh}(v/\mathrm{v}_\mathrm{0})$ respectively. The plant parameters are $(m_{[1]},m_{[2]},\mathrm{a},\mathrm{b},\mathrm{d},\mathrm{e},\mathrm{v}_\mathrm{0})=(0.1,0.01,1,1,0.5,0.5,0.01)$ in appropriate units. We obtain \eqref{eq:plant} by discretizing this system, with integration performed using the Runge-Kutta integrator (Tsit5) from the $\mathrm{Diffrax}$ library \cite{diffrax}, with a time-step of $0.02$s using inputs $u \in \mathbb{U}=[-1,1]$. We parameterize \eqref{eq:model} with $n_x=2$, $n_p=3$, and the FNN in \eqref{eq:softmax} with a single hidden layer containing $3$ $\mathrm{swish}$ activation units, such that $n_{\theta}=42$. Finally, we utilize $C=[1 \ 0]$ as the output matrix. For the TMPC controller, we set $F=[\mathbf{I}_{n_x} \ -\mathbf{I}_{n_x}]^{\top}$, resulting in $\fe=\ve=4$, and compute the matrix $E \in \R^{8 \times \fe}$ defining \eqref{eq:cc_constraint} using \cite{Villanueva2024}. We tune the tracking costs as $Q=\mathbb{I}_{n_x+n_u}$ and $P=Q/(1-\gamma^2)$, and the optimal RCI set costs in \eqref{eq:ell_formulation} as $Q_1=10$ and $Q_2=\mathrm{blkdiag}(10^{-6}\mathbb{I}_{n_x+n_u},10\mathbb{I}_{\fe},\mathbb{I}_{\ve n_u},\mathbb{I}_{\fe})$.

For our simulations, we initialize the model and plant at the origin, and compute $\hat{\theta}_0$ following \cite{mulagaleti2025} using an input-output dataset $\mathcal{D}_{\mathrm{train}}$ of length $T=100$ points such that Problem \eqref{eq:TMPC} is feasible at $(\hat{x}_0,\hat{\theta}_0)$ with $\beta=0.3,\gamma=0.95$ and $N=2$. Evaluating the model quality using the mean-squared error (MSE) loss $\ell_2=\sum_{t=0}^{T-1} \|y_t-\hat{y}_t\|_2^2/T$, $\hat{\theta}_0$ achieves MSE = $0.003$ over the training dataset. We also measure another dataset $\mathcal{D}_{\mathrm{test}}$ of length $10000$ points, over which $\hat{\theta}_0$ achieves MSE = $0.198$. In these datasets, we scale the plant output by $2.5$ to aid in effective learning. In the test set, the output tightly belongs in $[-1.487,1.329]$. Accordingly, we define $\mathbb{Y}=[-0.8,0.8]$, such that the output constraint set is smaller than the reachable set. We use the EKF in \eqref{eq:EKF_steps}-\eqref{eq:correction_qp_con} as the estimator in \eqref{eq:general_observer} with $Q_{\mathrm{e}}=\mathrm{blkdiag}(10^{-6}\mathbf{I}_{n_x},0\mathbf{I}_{n_{\theta}})$, $R_{\mathrm{e}}=0.1$ and $\hat{P}_0=\mathrm{blkdiag}(\mathbf{I}_{n_x},10^{-4}\mathbf{I}_{n_{\theta}})$. A small value of initial covariance is selected over $\hat{\theta}_0$ to avoid large variations $\|\hat{\theta}_{t+1}-\hat{\theta}_t\|$ during simulation.
As the active exploration criterion in Problem \eqref{eq:AL_problem}, we adopt the persistence-of-excitation condition for linear systems, given as
\begin{align}
	\label{eq:pe_condition}
	a(\mathbf{u})=\scalemath{1.}{\min \mathrm{eigval}\left(\sum_{j=t-T}^{t} u_{j:j+N}u_{j:j+N}^{T}\right)},
\end{align}
from \cite{Marafioti2013}, where $T \in \mathbb{N}$ is the horizon of past inputs, and $u_{j:j+N}=(u_j,\cdots,u_{j+N}) \in \R^{Nn_u}$ is the vector formed using the input sequence. We select $T=10$ in our simulations. As described in Section \ref{sec:persistence_excitation}, alternative criteria designed for nonlinear systems might result in improved performance. Finally, we solve Problem \eqref{eq:AL_problem} with $N^{\mathrm{p}}=N$ through enumeration using $M^{\mathrm{p}}=20$ samples defining \eqref{eq:sample_perturbations}.

In Figure \ref{fig:tracking_performance} (top), we demonstrate the performance of the closed-loop scheme in \eqref{eq:closed_loop_dual}, in which the output reference $\yr_t$ is a piecewise constant signal varying between $0.8$ and $-0.8$. 
By ``$\mathrm{No \ adapt}$'' we indicate the case when EKF is not used to estimate the state or the parameters, such that $\hat{x}_{t+1}=f(\hat{x}_t,u_t,\hat{\theta}_0)$, $\forall t\in\mathbb{N}$. We report no benefit if only the state is updated using the EKF with parameters $\hat{\theta}_t=\hat{\theta}_0$. Instead, when \eqref{eq:closed_loop_dual:2} is executed to perform joint state-parameter estimation, the performance improves over time: both $\beta=0$ (no additional perturbation) and $\beta=0.3$ result in improved tracking performance. Furthermore, faster parameter convergence is obtained with $\beta=0.3$, indicated by the model output $\hat{y}_t$ tracking the plant output $y_t$ earlier. 

As noted in Remark \ref{remark:constraints}, our scheme does not account for output constraint violation by the plant, as seen in Figure \ref{fig:tracking_performance}. However, the amount of violation reduces as the model quality improves. In Figure \ref{fig:tracking_performance} (middle), we see that the Lyapunov cost reduces when the reference is held constant, validating Theorem \ref{thm:stability}. Exponential decrease is observed since the tracking costs dominate, because $\|\hat{\theta}_{t+1}-\hat{\theta}_t\|$ is maintained small, as shown in Figure \ref{fig:tracking_performance} (bottom). In Figure \ref{fig:tracking_performance} (middle), we also plot optimality loss due to projection in the EKF, i.e., \eqref{eq:EKF_constraint}, where $\hat{\zeta}^{\mathrm{u}}_t$ denotes the unconstrained update following \eqref{eq:EKF_steps}, and the constrained update is computed following \cite{Simon2010} as $\hat{\zeta}_{t}=\arg\min_{\zeta \in \Theta_{t}} \|\zeta-\hat{\zeta}^{\mathrm{u}}_t\|_{\hat{P}_t^{-1}}^2$. Finally, we report that the average computation time per step is $8$~ms using the $\mathrm{qpax}$ QP-solver \cite{tracy2024differentiability} implemented in $\mathrm{jax}$ \cite{jax2018github} on a Windows 11 laptop with Intel Core i9-14900HX processor and $32$GB of RAM, indicating real-time implementability.

\begin{figure}[t]
	\centering
	\includegraphics[width=\linewidth, trim=0 0 0 0, clip]{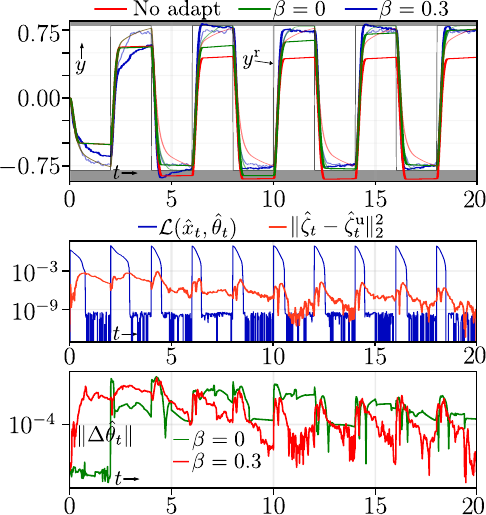}
	\caption{(Top) Tracking performance of \eqref{eq:closed_loop_dual}. The dull red, green and blue lines indicate the output trajectory $\hat{y}_t$ of \eqref{eq:model}, and the gray region indicates constraints; (Middle) Lyapunov cost defined in Theorem \ref{thm:stability}, and optimality loss due to projection \eqref{eq:EKF_constraint}; (Bottom) Variation in estimated parameter $\Delta \hat{\theta}_t=\hat{\theta}_{t+1}-\hat{\theta}_t$.}
	\label{fig:tracking_performance}
\end{figure}

We further analyze the performance of \eqref{eq:closed_loop_dual} over $10$ randomly generated piecewise constant reference trajectory realizations of $5000$ time steps. Using the plant output, we evaluate the tracking cost as $\mathrm{t}=\sum_{t=0}^{5000} \|y_t-\yr_t\|_2^2$, and constraint violation as $\mathrm{c}=\sum_{t=0}^{5000} \norm{\max_{i=1,2}\{H_iy_t-h_i,0\}^2}_1$, where $\mathbb{Y}=\left\{y \middle| Hy \leq h\right\}$. Finally, setting $\hat{\theta}^*=\hat{\theta}_{5000}$ as the identified model, we compute the MSE over the test dataset to evaluate model quality. In Table \ref{tab:adaptive_performance}, we present the mean and standard deviation statistics of these metrics over the reference trajectory realizations. First, we run the closed-loop without model adaptation. As expected, this results in a high tracking cost since the initial model does not capture the plant dynamics well. Then, adapting the model with \eqref{eq:general_observer}, we see that as the value of the active exploration constant $\beta$ increases, the tracking performance, constraint satisfaction, and final model quality improve. For small values of $\beta$, the constraint violation is worse since the system is excited beyond the output bounds before the model converges. For $\beta=0.3$, the fast model convergence results in the best overall tracking performance along with reduced constraint violation. We report that for $\beta>0.3$, the model output fails to track the reference, resulting in worse performance.

\section{Conclusions}
We have presented a dual MPC framework for nonlinear systems based on qLPV models, which is composed of a constrained estimator for state and parameters of the model, a TMPC controller, and a mechanism to perturb the inputs for the dual-control effect. The TMPC scheme is designed using configuration-constrained polytopes, and is guaranteed to be recursively feasible (Lemma \ref{lemma:recursive_feasibility}) and ISS-stable (Theorem \ref{thm:stability}) against parameter updates. Through a simple numerical example, the effectiveness of the framework is demonstrated using an EKF for parameter estimation, and active exploration criterion chosen to be the persistency-of-excitation condition.
Future research will focus on the incorporation of constraint tightening techniques to provide output constraint satisfaction guarantees, the analysis of suboptimality induced by the two-stage approach, and convergence of the constrained estimator scheme for qLPV systems.

\begin{table}[t]
	\centering
	\begin{threeparttable}
		\begin{tabular}{@{}lccc@{}}
			\toprule
			\textbf{Strategy} & \textbf{$\mathrm{t}$} & \textbf{$\mathrm{c}$} & \textbf{$\ell^{\mathrm{test}}_2(\hat{\theta}^*)$} \\ \midrule
			$\mathrm{No \ adapt}$ & $389.840 \pm 59.262$ & $0.146 \pm 0.208$ & $0.198 \pm 0.000$ \\
			$\beta=0$ & $350.037 \pm 76.337$ & $1.164 \pm 1.378$ & $0.045 \pm 0.037$ \\
			$\beta=0.1$ & $307.916 \pm 60.129$ & $3.668 \pm 3.617$ & $0.009 \pm 0.003$ \\
			$\beta=0.2$ & $243.600 \pm 37.850$ & $0.929 \pm 1.003$ & $0.007 \pm 0.002$ \\
			$\beta=0.3$ & $203.995 \pm 27.431$ & $0.180 \pm 0.230$ & $0.006 \pm 0.001$ \\ \bottomrule
		\end{tabular}
		\caption{Performance statistics (Mean and standard deviation) of \eqref{eq:closed_loop_dual} over different reference realization.}
		\label{tab:adaptive_performance}
	\end{threeparttable}
\end{table}

\bibliographystyle{ieeetr}
\bibliography{ieee_references}

\end{document}